\documentclass[journal,twoside,web,xcdraw,svgnames]{ieeecolor}
\usepackage{generic}
\usepackage{cite}
\usepackage{amsmath,amssymb,amsfonts}

\usepackage{hyperref}

\usepackage[T1]{fontenc}
\usepackage[utf8]{inputenc}
\usepackage{graphicx}
\usepackage{amsmath,systeme,amssymb,amsfonts}
\usepackage[version=4]{mhchem}
\let\labelindent\relax
\usepackage{siunitx}
\usepackage{enumitem}
\usepackage{longtable,tabularx}
\usepackage{comment}
\usepackage{textgreek}
\usepackage{float}
\usepackage[utf8]{inputenc}
\usepackage{tikz, pgfplots}
\usepackage{textcomp}
\usepackage{tcolorbox} % for colored boxes
\usepackage{graphicx}
\usepackage{amsmath}

\usepackage[version=4]{mhchem}
\usepackage{siunitx}
\usepackage{longtable,tabularx}
\usetikzlibrary{positioning, shapes, arrows}
\usepackage{thmtools}
\usepackage{textcomp}
\usepackage{soul}
\usepackage{float,amsfonts,amsthm,color,xcolor}
\usepackage{amssymb}
\usepackage{graphics} % for pdf, bitmapped graphics files
\usepackage{booktabs} % For formal tables
\usepackage{caption, subcaption}
\usepackage{breqn}
\usepackage{array, mathtools}
\usepackage{pgfplots}
\usepackage{siunitx}
\usepackage{algorithm}
\usepackage[noend]{algpseudocode}
\usepackage{tabulary}
\usetikzlibrary{arrows.meta}
\usepackage{stmaryrd} % for double brackets
\usepackage{varwidth}
\usepackage{comment}
\usetikzlibrary{positioning}

\usepackage{amsthm} \theoremstyle{definition} \newtheorem{theorem}{Theorem}  \newtheorem{assumption}{Assumption} \newtheorem{proposition}{Proposition} \newtheorem{lemma}{Lemma} \newtheorem{remark}{Remark} \newtheorem{corollary}{Corollary}

\def\BibTeX{{\rm B\kern-.05em{\sc i\kern-.025em b}\kern-.08em
    T\kern-.1667em\lower.7ex\hbox{E}\kern-.125emX}}
\title{Reward-Rate Congestion Games and Replicator--Dinkelbach Dynamics}
\author{Hassan Abdelraouf, Vaibhav Srivastava, and  Vijay Gupta
\thanks{H. Abdelraouf and V. Gupta are with the Elmore Family School of Electrical and Computer Engineering, Purdue University, West Lafayette, IN 47907, USA. (e-mails: abdelra5@purdue.edu, gupta869@purdue.edu).  V. Srivastava is with the Electrical and Computer Engineering, Michigan State University, East Lansing, MI 48824.(e-mail: vaibhav@msu.edu).}}

\begin{document}

\maketitle

\begin{abstract}

Reward rate is a key performance criterion in cyber-physical and robotic systems where time, workload, and coordination costs are limiting resources. We introduce reward-rate congestion games, where agents seek to maximize reward per unit execution time. The direct reward-rate game is generally not an exact potential game. We develop a Dinkelbach-based framework in which, for every fixed Dinkelbach parameter, the transformed game is an exact potential game. This yields a potential-level Dinkelbach iteration that terminates finitely at the optimal potential reward rate when the inner potential maximization problem is solved globally. We also provide a sufficient condition under which an equilibrium of the transformed game is an equilibrium of the original reward-rate game. To optimize aggregate performance, we introduce marginal externality corrections that make the corrected potential coincide with the Dinkelbach-transformed social reward-rate objective, thereby enabling optimization of the social reward rate. Finally, we develop a continuous-time replicator--Dinkelbach dynamics for reward-rate population games coupling fast replicator dynamics with a slow reward-rate update. We establish convergence of the fixed-parameter replicator dynamics,
global asymptotic and local exponential stability of the reduced Dinkelbach dynamics, and local exponential stability of the coupled system for sufficiently slow Dinkelbach updates. The framework is illustrated on a continuous task-allocation problem.
\end{abstract}

\begin{IEEEkeywords}
Reward-rate congestion games, potential games, replicator dynamics, task allocation.
\end{IEEEkeywords}
%%%%%%%%%%%%%%%%%%%%%%%%%%%%%%%%%%%%%%%%%%%%%%%%%%%%
\section{Introduction} \label{sec:introduction}

Congestion games provide a fundamental model for multi-agent decision-making over shared resources. 
%In these games, each agent selects a subset of resources, and the reward or cost associated with each resource depends on the number of agents using it. 
Since their introduction by Rosenthal~\cite{rosenthal1973class}, congestion games have been used to model 
%a wide range of networked systems, including 
routing, electrical grids, and network resource allocation problems~\cite{fotakis2002structure,ibars2010distributed,johari2004efficiency}. A central feature of congestion games is their connection to potential games: unilateral payoff variations can be represented by variations of a common potential function~\cite{monderer1996potential}. This potential structure
guarantees the existence of pure Nash equilibria and provides a natural basis for equilibrium computation and learning.

In many cyber-physical and robotic systems, however, performance is more naturally measured by reward per unit execution time than by total reward. Examples include search-and-rescue, environmental monitoring, warehouse automation, and inspection, where performance is measured by rates of detection, coverage, or task completion. Similar considerations arise for human--robot teams, where task decisions must account for robot execution time, human workload, and coordination costs. Thus, high total reward alone is insufficient when achieving it requires excessive time or human effort.

Reward-rate objectives introduce a structural difficulty. Standard congestion games are compatible with exact-potential methods because additive resource payoffs aggregate into a Rosenthal-type potential. In contrast, reward-rate payoffs consider the ratio of reward to execution time. Such ratio payoffs do not generally preserve the standard exact-potential structure of
congestion games. Consequently, decentralized learning and equilibrium computation methods for potential games cannot be applied directly to reward-rate congestion games.

This loss of potential structure is particularly important for learning and equilibrium computation, since the common potential function provides a global objective for analyzing many decentralized adjustment dynamics \cite{durand2018analysis,swenson2018best,heliou2017learning,cheung2020chaos,leslie2005individual,cominetti2010payoff,coucheney2015penalty}. It therefore motivates a transformation that recovers potential structure
while retaining the reward--time tradeoff. To address this challenge, we use the Dinkelbach transformation~\cite{dinkelbach1967nonlinear}. For each fixed value of the Dinkelbach parameter, the transformation replaces the reward-rate ratio by a reward-minus-time objective that is additive across resources, thereby recovering an exact-potential game. This structure enables a potential-level fractional optimization framework and, with appropriate externality corrections, optimization of the social reward rate.

%To address this challenge, this paper develops a Dinkelbach-based framework for reward-rate congestion games and their learning dynamics. 
The main contributions
of this paper are as follows. First, we introduce reward-rate congestion games,
where agents seek to maximize reward per unit execution time, and show that the
Dinkelbach-transformed game is an exact potential game for every fixed Dinkelbach parameter. This yields a potential-level Dinkelbach iteration
that terminates finitely at the optimal potential reward rate whenever the inner
problem is solved globally. Second, we introduce marginal externality
corrections that make the potential of the corrected transformed game
coincide with the Dinkelbach-transformed social reward-rate objective. Thus, optimizing the corrected potential optimizes the social reward rate. Third, we develop a continuous-time replicator--Dinkelbach dynamics for reward-rate population games. We establish convergence of the fixed-parameter replicator dynamics, global asymptotic and local exponential stability of the reduced Dinkelbach dynamics, and local exponential stability of the coupled system for sufficiently slow Dinkelbach updates. The framework is illustrated on a continuous population task-allocation problem.

%%%%%%%%%%%%%%%%%%%%%%%%%%%%%%%%%%%%%%%%%%%%
\section{Reward-Rate Congestion Games}\label{sec:reward-rate-congestion-games}

Consider a finite set of agents \(\mathcal N=\{1,\dots,N\}\) and a finite set of resources \(\mathcal E\). Each agent \(i\in\mathcal N\) selects an action
\(a_i\in\mathcal A_i\), where \(\mathcal A_i\subseteq 2^{\mathcal E}\). Let
\(
    a=(a_1,\dots,a_N)\in \mathcal A:=\prod_{i=1}^N \mathcal A_i
\)
denote an action profile. For each resource \(e\in\mathcal E\), define the corresponding congestion level by
\(
    n_e(a)
    :=
    \left|\{i\in\mathcal N: e\in a_i\}\right|.
\)
When resource \(e\) is used by \(k\) agents, it generates reward \(r_e(k)\) and incurs execution time \(t_e(k)>0\). Hence, the reward and execution time of agent
\(i\) under profile \(a\) are
\[
    R_i(a)
    =
    \sum_{e\in a_i} r_e(n_e(a)),
  \quad \text{and} \quad
    T_i(a)
    =
    \sum_{e\in a_i} t_e(n_e(a)).
\]
We assume that every feasible action is nonempty, i.e., \(a_i\neq\emptyset\) for all \(a_i\in\mathcal A_i\). Hence, \(T_i(a)>0\) for every agent \(i\) and every feasible profile \(a\). The corresponding reward-rate payoff is
\(
    J_i(a)
    =
    \frac{R_i(a)}{T_i(a)}.
\)
Thus, each agent seeks to maximize reward per unit execution time. We refer to the resulting game as a reward-rate congestion game.

The ratio form of \(J_i\) does not, in general, admit an exact potential.
Motivated by the Dinkelbach transformation, for each fixed
\(\rho\in\mathbb R\), the transformed payoff
\(
    Q_i^\rho(a)
    =
    R_i(a)-\rho T_i(a).
\)
Equivalently,
\(
    Q_i^\rho(a)
    =
    \sum_{e\in a_i}
    \left[
        r_e(n_e(a))-\rho t_e(n_e(a))
    \right].
\)
Define the transformed resource payoff
\(
    m_e^\rho(k)
    :=
    r_e(k)-\rho t_e(k).
\)
Then
\(
    Q_i^\rho(a)
    =
    \sum_{e\in a_i} m_e^\rho(n_e(a)).
\)

Recall that a game is an exact potential game if every unilateral change in a player’s payoff equals the corresponding change in a common potential function~\cite{monderer1996potential}.
For fixed \(\rho\),  define the Rosenthal-type potential 
\[
    \Phi^\rho(a)
    =
    \sum_{e\in\mathcal E}
    \sum_{q=1}^{n_e(a)}
    m_e^\rho(q).
\]
Equivalently,
\(
    \Phi^\rho(a)
    =
    \Phi_R(a)-\rho \Phi_T(a),
\)
where
\[
    \Phi_R(a)
    =
    \sum_{e\in\mathcal E}
    \sum_{q=1}^{n_e(a)} r_e(q),
    \qquad
    \Phi_T(a)
    =
    \sum_{e\in\mathcal E}
    \sum_{q=1}^{n_e(a)} t_e(q).
\]

\begin{proposition}
For every fixed \(\rho\in\mathbb R\), the transformed game with payoffs
\(\{Q_i^\rho\}_{i\in\mathcal N}\) is an exact potential game with potential
\(\Phi^\rho\).
\end{proposition}

\begin{proof}
Fix an agent \(i\), an action profile \(a=(a_i,a_{-i})\), and an alternative
action \(a_i'\in\mathcal A_i\). Let \(a'=(a_i',a_{-i})\). Only resources that are added or removed by player \(i\)'s deviation can change
their congestion levels. Define $ S^+ := a_i'\setminus a_i$ and $S^- := a_i\setminus a_i'$.
For \(e\in S^+\), the congestion level increases from \(n_e(a)\) to
\(n_e(a)+1\). For \(e\in S^-\), it decreases from \(n_e(a)\) to \(n_e(a)-1\).
Therefore,
\[
    \Phi^\rho(a')-\Phi^\rho(a)
    =
    \sum_{e\in S^+} m_e^\rho(n_e(a)+1)
    -
    \sum_{e\in S^-} m_e^\rho(n_e(a)).
\]
Similarly,
\[
    Q_i^\rho(a')-Q_i^\rho(a)
    =
    \sum_{e\in S^+} m_e^\rho(n_e(a)+1)
    -
    \sum_{e\in S^-} m_e^\rho(n_e(a)).
\]
Hence,
\[
    \Phi^\rho(a')-\Phi^\rho(a)
    =
    Q_i^\rho(a')-Q_i^\rho(a),
\]
which proves that \(\Phi^\rho\) is an exact potential.
\end{proof}

Thus, although the original reward-rate payoffs \(J_i=R_i/T_i\) are generally
non-potential, the Dinkelbach-transformed payoffs \(Q_i^\rho=R_i-\rho T_i\)
define an exact potential game for every fixed \(\rho\).

We now define the potential-level reward-rate problem associated with the transformed potential. Since every feasible action is nonempty and \(t_e(k)>0\) for all \(e\in\mathcal E\) and all \(k\ge 1\), we have
\(
    \Phi_T(a)>0
\)
for all \(a\in\mathcal A\). The potential-level reward-rate objective is
\begin{equation}\label{eq:potential_reward_rate_problem}
    \rho_\Phi^\star
    =
    \max_{a\in\mathcal A}
    \frac{\Phi_R(a)}{\Phi_T(a)}.
\end{equation}
By the Dinkelbach transformation~\cite{dinkelbach1967nonlinear}, the optimal ratio is the unique zero of the associated parametric value function
\[
    \psi_\Phi(\rho)
    =
    \max_{a\in\mathcal A}
    \left\{
        \Phi_R(a)-\rho \Phi_T(a)
    \right\}
    =
    \max_{a\in\mathcal A}
    \Phi^\rho(a).
\]
Then \(\rho_\Phi^\star\) is the unique root of \(\psi_\Phi\), i.e.,
\(
    \psi_\Phi(\rho_\Phi^\star)=0,
\)
and
\[
    \rho<\rho_\Phi^\star
    \Rightarrow
    \psi_\Phi(\rho)>0,
    \qquad
    \rho>\rho_\Phi^\star
    \Rightarrow
    \psi_\Phi(\rho)<0.
\]

The corresponding potential-level Dinkelbach iteration is
\begin{equation}\label{eq:inner_max}
     a^k
    \in
    \arg\max_{a\in\mathcal A}
    \Phi^{\rho_k}(a),
\end{equation}
followed by
\begin{equation}\label{eq:outer_update}
    \rho_{k+1}
    =
    \frac{\Phi_R(a^k)}{\Phi_T(a^k)}.
\end{equation}

For each fixed \(\rho_k\), the profile \(a^k\) is a global maximizer of the exact
potential \(\Phi^{\rho_k}\). Hence, \(a^k\) is a pure Nash equilibrium of the
transformed game with payoffs \(\{Q_i^{\rho_k}\}_{i\in\mathcal N}\). However, a
pure Nash equilibrium of the transformed game need not be a global maximizer of
\(\Phi^{\rho_k}\). Therefore, convergence of the Dinkelbach iteration requires
the inner potential maximization problem to be solved globally.

\begin{theorem}\label{thrm:convergence_dinklach}
Suppose \(\Phi_T(a)>0\) for all \(a\in\mathcal A\), and initialize
\(
    \rho_0=\frac{\Phi_R(a^0)}{\Phi_T(a^0)}
\)
for some \(a^0\in\mathcal A\). Suppose that, at each iteration, the inner problem \eqref{eq:inner_max} is solved globally. Then the sequence \(\{\rho_k\}\) generated by \eqref{eq:inner_max}--\eqref{eq:outer_update} is nondecreasing and reaches \(\rho_\Phi^\star\) in finitely many iterations, where \(\rho_\Phi^\star\) is defined
in \eqref{eq:potential_reward_rate_problem}. The corresponding profile \(a^\star\)
globally solves \eqref{eq:potential_reward_rate_problem}.
\end{theorem}

\begin{proof}
Since \(\rho_0=\Phi_R(a^0)/\Phi_T(a^0)\) is a feasible ratio, we have
\(
    \rho_0\le \rho_\Phi^\star.
\)
At iteration \(k\), because \eqref{eq:inner_max} is solved globally, we have
\(
    \psi_\Phi(\rho_k)
    =
    \Phi_R(a^k)-\rho_k\Phi_T(a^k).
\)
Using \eqref{eq:outer_update}, we obtain
\(
    \rho_{k+1}
    =
    \frac{\Phi_R(a^k)}{\Phi_T(a^k)}
    =
    \rho_k
    +
    \frac{\psi_\Phi(\rho_k)}{\Phi_T(a^k)}.
\)
Since \(\Phi_T(a^k)>0\), if \(\rho_k<\rho_\Phi^\star\), then the sign property of
\(\psi_\Phi\) gives
\(
    \psi_\Phi(\rho_k)>0,
\)
and therefore
\(
    \rho_{k+1}>\rho_k.
\)
Moreover, \(\rho_{k+1}\) is a feasible ratio, so
\(
    \rho_{k+1}\le \rho_\Phi^\star.
\)
Hence, \(\{\rho_k\}\) is nondecreasing and bounded above by
\(\rho_\Phi^\star\).

Since \(\mathcal A\) is finite, the set of feasible ratios
\(
    \left\{
    \frac{\Phi_R(a)}{\Phi_T(a)}:a\in\mathcal A
    \right\}
\)
is finite. Therefore, the nondecreasing sequence \(\{\rho_k\}\), whose elements
belong to this finite set, must terminate in finite time. Let \(\rho_K\) be its
terminal value. If \(\rho_K<\rho_\Phi^\star\), then the preceding argument implies
\(    \rho_{K+1}>\rho_K,
\)
which contradicts termination. Hence,
\(
    \rho_K=\rho_\Phi^\star.
\)
Finally, the terminal profile \(a^\star:=a^K\) satisfies
\(
    \frac{\Phi_R(a^\star)}{\Phi_T(a^\star)}
    =
    \rho_\Phi^\star,
\)
and therefore globally solves \eqref{eq:potential_reward_rate_problem}.
\end{proof}

Since the terminal profile \(a^\star\) globally maximizes the exact potential \(\Phi^{\rho_\Phi^\star}\), it is also a pure Nash equilibrium of the transformed game with payoffs
\(
Q_i^{\rho_\Phi^\star}=R_i-\rho_\Phi^\star T_i.
\)
The next result gives a sufficient condition under which an equilibrium of the transformed game is also an equilibrium of the original reward-rate game.

\begin{corollary}\label{col:same_NE}
Let \(\bar\rho\in\mathbb R\), and let \(a^\star\) be a pure Nash equilibrium of the transformed game with payoffs
\(
Q_i^{\bar\rho}(a)=R_i(a)-\bar\rho T_i(a).
\)
If
\(
    \frac{R_i(a^\star)}{T_i(a^\star)}
    =
    \bar\rho,
    \; \forall i\in\mathcal N,
\)
then \(a^\star\) is a pure Nash equilibrium of the direct reward-rate game
with payoffs
\(
J_i(a)=\frac{R_i(a)}{T_i(a)}.
\)
\end{corollary}

\begin{proof}
Since \(a^\star\) is a pure Nash equilibrium of the transformed game, for every
player \(i\) and every unilateral deviation \(a_i'\in\mathcal A_i\),
\[
    R_i(a^\star)-\bar\rho T_i(a^\star)
    \ge
    R_i(a_i',a_{-i}^\star)
    -
    \bar\rho T_i(a_i',a_{-i}^\star).
\]
By the assumed common-rate condition,
\(
R_i(a^\star)-\bar\rho T_i(a^\star)=0.
\)
Therefore,
\(
    R_i(a_i',a_{-i}^\star)
    -
    \bar\rho T_i(a_i',a_{-i}^\star)
    \le 0.
\)
Since \(T_i(a_i',a_{-i}^\star)>0\), it follows that
\(
    \frac{R_i(a_i',a_{-i}^\star)}
    {T_i(a_i',a_{-i}^\star)}
    \le
    \bar\rho.
\)
Using 
\(
    \bar\rho
    =
    \frac{R_i(a^\star)}{T_i(a^\star)}
\), we get 
\(
    J_i(a_i',a_{-i}^\star)
    \le
    J_i(a^\star),
    \; \forall a_i'\in\mathcal A_i.
\)
Thus no player can improve its direct reward rate by a unilateral deviation, and \(a^\star\) is a pure Nash equilibrium of the direct reward-rate game.
\end{proof}
Thus, an equilibrium of the transformed game is not, in general, a Nash equilibrium of the direct reward-rate game; Corollary \ref{col:same_NE} provides a sufficient condition under which the same profile is a Nash equilibrium of both games.
%%%%%%%%%%%%%%%%%%%%%%%%%%%%%%%%%%%%%%%%%%%%%%%%%%%%%%%%%
\section{Social-Reward-Rate  
Alignment via Externality Corrections}\label{sec:social-reward-rate-alignment}

The potential-level reward-rate objective in \eqref{eq:potential_reward_rate_problem}
does not generally coincide with the social reward-rate objective. Indeed, the aggregate social reward and aggregate social execution time are
\begin{align*}
 SW_R(a)
    & =
    \sum_{i=1}^N R_i(a)
    =
    \sum_{e\in\mathcal E} n_e(a) r_e(n_e(a)),  \quad \text{and} \\
     SW_T(a)
    &=
    \sum_{i=1}^N T_i(a)
    =
    \sum_{e\in\mathcal E} n_e(a) t_e(n_e(a)).
\end{align*}
In general,
\(
    \Phi_R(a)\neq SW_R(a),
    \;
    \Phi_T(a)\neq SW_T(a).
\)
Thus, optimizing the potential-level ratio
\(
    \frac{\Phi_R(a)}{\Phi_T(a)}
\)
does not necessarily optimize the social reward rate
\(
    \frac{SW_R(a)}{SW_T(a)}.
\)

To align transformed individual incentives with the social reward-rate objective, we assign each resource its marginal contribution to aggregate reward and execution time. Specifically, define

\begin{align*}
   \widetilde r_e(k)
    &:=
    k r_e(k)-(k-1)r_e(k-1),  \quad \text{and} \\
     \widetilde t_e(k)
    &:=
    k t_e(k)-(k-1)t_e(k-1),
\end{align*}
with the convention
\(
    r_e(0)=0 \) and \(
    t_e(0)=0.
\)
The corrected payoff of agent \(i\) is
\(
    \widetilde R_i(a)
    =
    \sum_{e\in a_i}\widetilde r_e(n_e(a))\) 
    and \(
    \widetilde T_i(a)
    =
    \sum_{e\in a_i}\widetilde t_e(n_e(a)).
\)
For fixed \(\rho\), define the corrected transformed payoff
\(
    \widetilde Q_i^\rho(a)
    =
    \widetilde R_i(a)-\rho \widetilde T_i(a).
\)

\begin{theorem}
For every fixed \(\rho\), the corrected transformed game with payoffs \(\{\widetilde Q_i^\rho\}_{i\in\mathcal N}\) is an exact potential game with potential
\(
    \widetilde \Phi^\rho(a)
    =
    SW_R(a)-\rho SW_T(a).
\)
\end{theorem}

\begin{proof}
By construction,
\[
    \sum_{q=1}^{k}\widetilde r_e(q)
    =
    \sum_{q=1}^{k}
    \bigl[q r_e(q)-(q-1)r_e(q-1)\bigr]
    =
    k r_e(k),
\]
and similarly,
\(
    \sum_{q=1}^{k}\widetilde t_e(q)
    =
    k t_e(k).
\)

Hence, the Rosenthal reward and execution-time potentials of the corrected game satisfy
\begin{align*}
  \widetilde\Phi_R(a)
    &=
    \sum_{e\in\mathcal E} n_e(a)r_e(n_e(a))
    =
    \mathrm{SW}_R(a),  \quad \text{and} \\
    \widetilde\Phi_T(a)
    &=
    \sum_{e\in\mathcal E} n_e(a)t_e(n_e(a))
    =
    \mathrm{SW}_T(a).
\end{align*}
Therefore,
\[
    \widetilde\Phi^\rho(a)
    =
    \widetilde\Phi_R(a)-\rho \widetilde\Phi_T(a)
    =
    \mathrm{SW}_R(a)-\rho \mathrm{SW}_T(a).
\]
The exact-potential property then follows from the same unilateral-deviation argument as in Proposition~1. 
\end{proof}

Consequently, since \(SW_T(a)>0\) for all feasible profiles, the social reward-rate problem
\[
    \rho_{SW}^\star
    =
    \max_{a\in\mathcal A}
    \frac{SW_R(a)}{SW_T(a)}
\]
admits the same potential-level Dinkelbach construction as in Section~\ref{sec:reward-rate-congestion-games}. Under global solution of each inner problem,
\[
    a^k
    \in
    \arg\max_{a\in\mathcal A}
    \left\{
        SW_R(a)-\rho_k SW_T(a)
    \right\}, \quad \rho_{k+1}
    =
    \frac{SW_R(a^k)}{SW_T(a^k)}
\]
reaches the optimal social reward rate in finitely many iterations.

% can be solved, under global inner maximization, through the corrected
% potential-level Dinkelbach iteration

% followed by
% \[
%     \rho_{k+1}
%     =
%     \frac{SW_R(a^k)}{SW_T(a^k)}.
% \]
%%%%%%%%%%%%%%%%%%%%%%%%%%%%%%%%%%%%%%%%%%%%%%%%%%%%%%%%%
\section{Task Allocation as a Singleton Congestion Game}

% Game-theoretic formulations have been used to model multi-robot task allocation, where agents repeatedly select tasks and adapt their allocations over time~\cite{park2021multi}. 
Task allocation\cite{park2021multi} provides a natural setting for the proposed reward-rate congestion-game framework. 
%We therefore specialize the preceding model to task-allocation problems. 
Let \(\mathcal T=\{1,\dots,M\}\) denote a finite set of tasks, and suppose each agent selects one task. The resulting game is a singleton congestion game in which tasks are the resources.
% Each agent selects one
% task, so that
% \(
%     a_i\in \mathcal T,
%     \; i\in\mathcal N .
% \)
% Thus, the task-allocation problem is a singleton congestion game in which tasks play the role of resources. 
% For a task \(j\in\mathcal T\), define the set of agents assigned to task \(j\) by
For a task \(j\), let 
\[
    S_j(a):=\{i\in\mathcal N:a_i=j\}, \quad \text{and} \quad  n_j(a):=|S_j(a)|
\]
denote the assigned coalition and its size.
% and let
% \(
%  n_j(a):=|S_j(a)|  
% \)
% denote the corresponding coalition size.

Suppose that task \(j\) generates individual reward \(r_j(k)\) and individual execution time \(t_j(k)>0\) when \(k\) agents are assigned to it. Then the
reward and execution time of agent \(i\), with \(a_i=j\), are
\(
    R_i(a)=r_j(n_j(a)), \) and \(
    T_i(a)=t_j(n_j(a)).
\)
% The direct reward-rate payoff is therefore
% \[
%     J_i(a)=\frac{r_{a_i}(n_{a_i}(a))}
%     {t_{a_i}(n_{a_i}(a))}.
% \]
%For fixed \(\rho\), the Dinkelbach-transformed payoff is
The corresponding fixed-parameter transformed payoff is
\[
    Q_i^\rho(a)
    =
    r_{a_i}(n_{a_i}(a))
    -
    \rho t_{a_i}(n_{a_i}(a)).
\]
By the previous results, the transformed task-allocation game is an exact
potential game with potential
\[
    \Phi^\rho(a)
    =
    \sum_{j=1}^M
    \sum_{q=1}^{n_j(a)}
    \left[
        r_j(q)-\rho t_j(q)
    \right].
\]

In the language of anonymous hedonic games \cite{jang2018anonymous}, the coalition assigned to task \(j\) is \(S_j(a)\), and each agent's payoff depends only on the selected task and the
size of its coalition. Thus, for both the direct reward-rate payoff and the fixed-parameter transformed payoff, the task-allocation model is equivalently a singleton congestion game. Nash stability of the allocation is therefore equivalent to pure Nash equilibrium: no agent can improve its payoff by unilaterally switching tasks.
%%%%%%%%%%%%%%%%%%%%%%%%%%%%%%%%%%%%%%%%%%%%%%%%%%%%%%%%%%%%%

\section{Continuous-Time Replicator--Dinkelbach Dynamics}

We now consider a continuous population formulation of the social reward-rate problem. Let
\(
    x=(x_1,\dots,x_M)\in\Delta_M,
\)
where \(x_j\) denotes the fraction of the population assigned to task \(j\). Let \(B_j(x_j)\) and \(H_j(x_j)\) denote the aggregate
%continuous counterparts of the total 
reward and %total 
execution time associated with task \(j\). Define
\[
    R(x)=\sum_{j=1}^M B_j(x_j),
    \qquad
    T(x)=\sum_{j=1}^M H_j(x_j),
\]
and assume \(T(x)>0\) for all \(x\in\Delta_M\). The continuous reward-rate problem is
\begin{equation}\label{eq:continuous_reward_rate}
    \rho^\star
    =
    \max_{x\in\Delta_M}
    \frac{R(x)}{T(x)}.
\end{equation}

For a fixed \(\rho\), define the Dinkelbach-transformed objective
\[
    W^\rho(x)
    =
    R(x)-\rho T(x).
\]
The marginal transformed payoff of task \(j\) is
\[
    u_j^\rho(x)
    :=
    \frac{\partial W^\rho(x)}{\partial x_j}
    =
    B_j'(x_j)-\rho H_j'(x_j).
\]
These marginal transformed payoffs are the population analogue of the externality-corrected payoffs in Section~\ref{sec:social-reward-rate-alignment}.
Let
\(
    \bar u^\rho(x)
    =
    \sum_{j=1}^M x_j u_j^\rho(x).
\)
For fixed \(\rho\), consider the replicator dynamics \cite{taylor1978evolutionary,hofbauer1998evolutionary}
\begin{equation}\label{eq:rd_inner}
    \dot x_j
    =
    x_j\left(u_j^\rho(x)-\bar u^\rho(x)\right),
    \qquad j=1,\dots,M.
\end{equation}

The dynamics \eqref{eq:rd_inner} is the inner, or fast, dynamics for the
Dinkelbach-transformed objective with \(\rho\) fixed. To solve the reward-rate
problem \eqref{eq:continuous_reward_rate}, we couple this inner dynamics with a
slow Dinkelbach update for \(\rho\):
\begin{equation}\label{eq:full_two_time_scale}
\begin{aligned}
    \dot x_j
    &=
    x_j\left(u_j^\rho(x)-\bar u^\rho(x)\right),
    \qquad j=1,\dots,M,\\
    \dot \rho
    &=
    \varepsilon
    \left(
        \frac{R(x)}{T(x)}-\rho
    \right),
    \qquad 0<\varepsilon\ll 1.
\end{aligned}
\end{equation}
The small parameter \(\varepsilon\) separates the time scales: \(x\) evolves on the fast time scale, while \(\rho\) evolves slowly according to the currently achieved reward rate \(R(x)/T(x)\). We first analyze the frozen-\(\rho\) inner dynamics \eqref{eq:rd_inner}, and then study the reduced slow dynamics induced by \eqref{eq:full_two_time_scale}.

\begin{assumption}\label{ass:regularity}
For every $j$, $B_j$ and $H_j$ are thrice continuously
differentiable on an open interval containing $[0,1]$. Moreover,
there exists $\mu>0$ such that $B_j''(s)\leq -\mu$ and $H_j''(s)\geq 0$, for each $ s\in[0,1]$, 
and $T(x)>0$, for every $x\in\Delta_M$. 
\end{assumption}
The above assumption implies $W^\rho$ is uniformly strongly concave and has a unique maximizer.

\begin{assumption}\label{ass:interiority}
There exists a compact interval
$\mathcal I\subset\mathbb R_{\geq0}$ containing $\rho^\star$
such that $ x^\star(\rho)\in\operatorname{int}(\Delta_M)$, for every $\rho\in\mathcal I$. 
\end{assumption}

\begin{lemma}\label{lem:optimizer_regular}
Under Assumptions~\ref{ass:regularity} and~\ref{ass:interiority},
the optimizer map $ \rho\mapsto x^\star(\rho)$
is twice continuously differentiable on $\mathcal I$.
\end{lemma}
\begin{proof}
Since $ x^\star(\rho)\in\operatorname{int}(\Delta_M)$, the KKT conditions for
maximizing $W^\rho$ over $\Delta_M$ are
\[
    B_j'(x_j^\star)
    -
    \rho H_j'(x_j^\star)
    =
    \lambda,
    \qquad j=1,\ldots,M,
\]
together with $\mathbf 1^\top x^\star=1$.
Define
\[
    F(x,\lambda,\rho)
    =
    \begin{bmatrix}
    B_1'(x_1)-\rho H_1'(x_1)-\lambda\\
    \vdots\\
    B_M'(x_M)-\rho H_M'(x_M)-\lambda\\
    \mathbf 1^\top x-1
    \end{bmatrix}.
\]
Its Jacobian with respect to $(x,\lambda)$ is $D_{(x,\lambda)}F
    =
    \begin{bmatrix}
        D & -\mathbf 1\\
        \mathbf 1^\top & 0
    \end{bmatrix}$, 
where $ D
    =
    \operatorname{diag}
    \left(
        B_j''(x_j)-\rho H_j''(x_j)
    \right)$.

By Assumption~\ref{ass:regularity}, $D\preceq-\mu I$. Consider a vector $\begin{bmatrix}
    v^\top \\ \alpha
\end{bmatrix}^\top$. Suppose
\[
    Dv-\alpha\mathbf 1=0,
    \quad \text{and} \quad 
    \mathbf 1^\top v=0.
\]
Multiplying the first equality by $v^\top$ gives $ v^\top Dv = \alpha\mathbf 1^\top v=    0$. 
Since $D$ is negative definite, $v=0$, and consequently
$\alpha=0$. Hence, $D_{(x,\lambda)}F$ is nonsingular.

The implicit-function theorem therefore implies that
$(x^\star(\rho),\lambda^\star(\rho))$ is twice continuously
differentiable locally in $\rho$. Since the maximizer is unique for every
\(\rho\in\mathcal I\), these local representations define a twice continuously
differentiable optimizer map \(x^\star(\rho)\) on \(\mathcal I\).
%Uniqueness of the maximizer
%ensures that $x^\star(\rho)$ is continuously differentiable on
%$\mathcal I$.
\end{proof}

It is well known that the simplex \(\Delta_M\) is forward invariant under
replicator dynamics \cite{hofbauer1998evolutionary}. Moreover, for fixed \(\rho\),
\[
\dot W^\rho(x(t))
    =
    \sum_{j=1}^M
    x_j
    \left(
        u_j^\rho(x)-\bar u^\rho(x)
    \right)^2
    \ge 0,
\]
so the transformed objective is nondecreasing along its trajectories. The next result characterizes the stability of the frozen-\(\rho\) fast
subsystem. 

\begin{theorem}\label{thm:fast}
Fix $\rho\geq0$. Under Assumption~\ref{ass:regularity}, every
solution of the frozen-$\rho$ replicator dynamics~\eqref{eq:rd_inner}
with an interior initial condition satisfies $ x(t)\longrightarrow x^\star(\rho)$. 
Thus, $x^\star(\rho)$ is globally asymptotically attractive
relative to $\operatorname{int}(\Delta_M)$.

Moreover, under Assumption~\ref{ass:interiority},
$x^\star(\rho)$ is locally exponentially stable for every
$\rho\in\mathcal I$, uniformly in $\rho$.
\end{theorem}

\begin{proof}
Fix \(\rho\) and write \(x^\star=x^\star(\rho)\). Consider
\[
    V(x)
    =
    \sum_{j=1}^M x_j^\star\log\frac{x_j^\star}{x_j}.
\]
Since \(x(0)\in\operatorname{int}(\Delta_M)\), the replicator dynamics preserves positivity, so \(x_j(t)>0\) for all \(j\) and all \(t\ge 0\). Hence the logarithmic terms in \(V(x(t))\) are finite, and \(V\) is well-defined along solutions.  Since \(V\) is the Kullback--Leibler divergence from \(x^\star\) to \(x\),
\(
    V(x)\ge 0,
\) with equality if and only if \(x=x^\star\).
Along \eqref{eq:rd_inner},
\[
    \dot V
    =
    -\sum_{j=1}^M x_j^\star\frac{\dot x_j}{x_j}
    =
    -\sum_{j=1}^M x_j^\star
    \left(u_j^\rho(x)-\bar u^\rho(x)\right).
\]
Since \(\sum_{j=1}^M x_j^\star=1\) and
\(
    \bar u^\rho(x)=\sum_{j=1}^M x_j u_j^\rho(x),
\)
we obtain
\[
    \dot V
    =
    -\sum_{j=1}^M x_j^\star u_j^\rho(x)
    +
    \sum_{j=1}^M x_j u_j^\rho(x).
\]
Therefore,
\(
    \dot V
    =
    -\nabla W^\rho(x)^\top(x^\star-x).
\)
By concavity of \(W^\rho\),
\[
    W^\rho(x^\star)-W^\rho(x)
    \le
    \nabla W^\rho(x)^\top(x^\star-x).
\]
Since \(x^\star\) maximizes \(W^\rho\), the left-hand side is nonnegative. Hence,
\[
    \dot V
    \le
    -\left(W^\rho(x^\star)-W^\rho(x)\right)
    \le 0.
\]
Define the optimality gap
\(
    g(t):=W^\rho(x^\star)-W^\rho(x(t))\ge 0.
\)
From
\(
    \dot V
    \le
    -g(t),
\)
it follows that
\[
    \int_0^T g(t)\,dt
    \le
    V(x(0))-V(x(T))
    \le
    V(x(0)).
\]
Therefore, $g(t)$ is integrable. 
% \[
%     \int_0^\infty g(t)\,dt<\infty.
% \]
Moreover, \(g(t)\) is uniformly continuous. Indeed,
\(
    \dot g(t)
    =
    -\nabla W^\rho(x(t))^\top \dot x(t).
\)
Since \(x(t)\in\Delta_M\), and \(\Delta_M\) is compact, \(\nabla W^\rho\) is
bounded on \(\Delta_M\). Also, the vector field in \eqref{eq:rd_inner} is
continuous on \(\Delta_M\), and hence \(\dot x(t)\) is bounded. Therefore,
\(\dot g(t)\) is bounded, which implies that \(g(t)\) is uniformly continuous. Hence, by Barbalat's lemma,
\(
    g(t)\to 0.
\)
Thus,
\(
    W^\rho(x(t))\to W^\rho(x^\star).
\)
Since \(W^\rho\) is strictly concave, \(x^\star\) is the unique maximizer, and therefore
\(
    x(t)\to x^\star.
\)

To establish the local exponential property, let
$\rho\in\mathcal I$. Uniform strong concavity gives
\[
    W^\rho(x^\star(\rho))-W^\rho(x)
    \geq
    \frac{\mu}{2}
    \|x-x^\star(\rho)\|^2.
\]
Hence, for the KL Lyapunov function used above,
\[
    \dot V
    \leq
    -\frac{\mu}{2}
    \|x-x^\star(\rho)\|^2.
\]
By Lemma~\ref{lem:optimizer_regular}, $x^\star(\rho)$ is
continuous in $\rho$. Since $\mathcal I$ is compact and
$x^\star(\rho)$ is interior, the equilibrium manifold is
uniformly separated from the boundary of the simplex.
Consequently, in a uniform neighborhood of this manifold,
\[
    c_1\|x-x^\star(\rho)\|^2
    \leq V(x)
    \leq
    c_2\|x-x^\star(\rho)\|^2,
\]
where $c_1,c_2>0$ are independent of $\rho$. Therefore,
$x^\star(\rho)$ is locally exponentially stable uniformly over
$\rho\in\mathcal I$.
\end{proof}

Theorem~\ref{thm:fast} therefore provides the boundary-layer stability
property required for the two-time-scale analysis. We next consider the
reduced slow dynamics obtained by restricting the fast state to the optimizer
manifold \(x=x^\star(\rho)\). Substituting \(x=x^\star(\rho)\) into
\eqref{eq:full_two_time_scale} gives
\begin{equation}\label{eq:reduced_rho}
    \dot\rho
    =
    \varepsilon\bigl(h(\rho)-\rho\bigr),
\end{equation}
where
\(
    h(\rho)
    :=
    \frac{R(x^\star(\rho))}{T(x^\star(\rho))}.
\)

\begin{theorem}\label{thm:reduced}
Under Assumptions~\ref{ass:regularity} and
\ref{ass:interiority}, the optimal value $\rho^\star$ in \eqref{eq:continuous_reward_rate} is globally
asymptotically stable and locally exponentially stable for the reduced dynamics \eqref{eq:reduced_rho}. 
\end{theorem}

\begin{proof}
Define
\[
    \psi(\rho)
    =
    \max_{x\in\Delta_M}
    \{R(x)-\rho T(x)\}.
\]
Since \(x^\star(\rho)\) maximizes \(R(x)-\rho T(x)\),
\[
    \psi(\rho)
    =
    R(x^\star(\rho))-\rho T(x^\star(\rho))
    =
    T(x^\star(\rho))(h(\rho)-\rho).
\]
Because \(T(x^\star(\rho))>0\), the sign of \(h(\rho)-\rho\) is the same as the
sign of \(\psi(\rho)\). By the Dinkelbach sign property,
\[
    \rho<\rho^\star \Rightarrow h(\rho)-\rho>0,
    \qquad
    \rho>\rho^\star \Rightarrow h(\rho)-\rho<0.
\]
Consider 
\(
    V_\rho(\rho)=\frac12(\rho-\rho^\star)^2.
\)
Along \eqref{eq:reduced_rho},
\[
    \dot V_\rho
    =
    \varepsilon(\rho-\rho^\star)\bigl(h(\rho)-\rho\bigr).
\]
If \(\rho>\rho^\star\), then \(h(\rho)-\rho<0\), and hence
\(\dot V_\rho<0\). If \(\rho<\rho^\star\), then \(h(\rho)-\rho>0\), and again
\(\dot V_\rho<0\). Therefore,
\(
    \dot V_\rho<0,
    \; \forall \rho\neq \rho^\star.
\)
Thus, \(\rho^\star\) is globally asymptotically stable for the reduced slow
dynamics.

It remains to establish local exponential stability. Define
\[
    \psi(\rho)
    =
    W^\rho(x^\star(\rho))
    =
    T(x^\star(\rho))(h(\rho)-\rho).
\]

Differentiating and using the KKT condition
$\nabla W^\rho(x^\star)=\lambda\mathbf 1$ gives
\[
\begin{aligned}
    \psi'(\rho)
    &=
    \nabla W^\rho(x^\star)^\top{x^\star}'(\rho)
    -T(x^\star(\rho))\\
    &=
    \lambda\mathbf1^\top{x^\star}'(\rho)
    -T(x^\star(\rho))
    =
    -T(x^\star(\rho)),
\end{aligned}
\]
where $\mathbf1^\top{x^\star}'(\rho)=0$ follows from
$\mathbf1^\top x^\star(\rho)=1$.

Since $\psi(\rho^\star)=0$,
\[
    \left.
    \frac{d}{d\rho}(h(\rho)-\rho)
    \right|_{\rho=\rho^\star}
    =
    \frac{\psi'(\rho^\star)}
         {T(x^\star(\rho^\star))}
    =
    -1.
\]
Therefore, the linearization of the reduced dynamics at
$\rho^\star$ is
\(
    \dot{\widetilde\rho}
    =
    -\varepsilon\widetilde\rho,
\)
and $\rho^\star$ is locally exponentially stable.
\end{proof}

The preceding results provide the fast- and slow-subsystem stability properties required for the singular-perturbation analysis.
\begin{theorem}[Two-time-scale exponential stability]
\label{thm:two_time_scale}
Suppose Assumptions~\ref{ass:regularity} and
\ref{ass:interiority} hold. Then there exists
$\varepsilon^\star>0$ such that, for every
$0<\varepsilon<\varepsilon^\star$, $\bigl(x^\star(\rho^\star),\rho^\star\bigr)$
is a locally exponentially stable equilibrium of the coupled
replicator--Dinkelbach dynamics~\eqref{eq:full_two_time_scale}.
\end{theorem}

\begin{proof}
Using the slow time $s=\varepsilon t$, the coupled dynamics has the
standard singularly perturbed form with slow variable $\rho$ and
fast variable $x$. By Lemma~\ref{lem:optimizer_regular}, the
isolated equilibrium $x^\star(\rho)$ of the fast subsystem
depends smoothly on $\rho$. Under
Assumption~\ref{ass:regularity}, the system functions and their
required derivatives are locally bounded.

By Theorem~\ref{thm:reduced}, the equilibrium $\rho^\star$ of
the reduced system is locally exponentially stable. By
Theorem~\ref{thm:fast}, the equilibrium $x^\star(\rho)$ of the
boundary-layer system is locally exponentially stable uniformly
in $\rho\in\mathcal I$. Therefore, the hypotheses of
\cite[Thm.~11.4]{khalil2002nonlinear} are satisfied, and the
result follows.
\end{proof}

Thus, for sufficiently small \(\varepsilon\), time-scale separation preserves
the local exponential stability of the reward-rate optimum: the fast
replicator dynamics tracks the optimizer \(x^\star(\rho)\), while the slow
Dinkelbach update drives \(\rho\) toward \(\rho^\star\).

\begin{remark}
The analysis is not specific to replicator dynamics. The same
singular-perturbation argument applies to any sufficiently smooth fast
learning dynamics whose equilibrium \(x^\star(\rho)\) solves the
fixed-\(\rho\) potential problem and is uniformly locally exponentially
stable in \(\rho\in\mathcal I\).
\end{remark}

\section{Numerical Example}
We illustrate the proposed replicator--Dinkelbach dynamics on a continuous task-allocation problem with three tasks: victim search, damage inspection, and communication relay/mapping. For each task \(j\in\{1,2,3\}\), let
\[
    B_j(x_j)=\alpha_j(1-e^{-\beta_j x_j}),
    \qquad
    H_j(x_j)=c_jx_j+d_jx_j^2,
\]
where \(x_j\) is the fraction of the population assigned to task \(j\). The
function \(B_j\) models saturating task reward, while \(H_j\) models increasing
execution-time burden due to congestion. We use the parameters
\[
\begin{array}{c|cccc}
\text{task }j & \alpha_j & \beta_j & c_j & d_j \\ \hline
1 & 10   & 4.0 & 1.2 & 3.0 \\
2 & 7    & 2.5 & 0.8 & 1.2 \\
3 & 5.5  & 1.5 & 0.5 & 0.5
\end{array}
\]
and define
\(
    R(x)=\sum_{j=1}^3 B_j(x_j),
    \;
    T(x)=\sum_{j=1}^3 H_j(x_j).
\)
The objective is
\(
    \rho^\star
    =
    \max_{x\in\Delta_3}
    {R(x)}/{T(x)}.
\)
For a fixed \(\rho\), the marginal transformed payoff is
\[
    u_j^\rho(x)
    =
    \alpha_j\beta_j e^{-\beta_jx_j}
    -
    \rho(c_j+2d_jx_j).
\]
We simulate the replicator--Dinkelbach dynamics
\eqref{eq:full_two_time_scale} with \(\varepsilon=0.05\).
As a reference solution, we first compute the optimal reward-rate pair \((x^\star,\rho^\star)\) using the static Dinkelbach iteration. Starting from \(\rho_0\ge 0\), each iteration solves
\(
    x^k
    \in
    \arg\max_{x\in\Delta_3}
    \left\{
        R(x)-\rho_k T(x)
    \right\},
\)
and updates
\(
    \rho_{k+1}
    =
    {R(x^k)}/{T(x^k)}.
\)
Since \(R\) is strictly concave and \(T\) is strictly convex, \(R-\rho_kT\) is strictly concave for every \(\rho_k\ge 0\). Hence, the inner problem has a
unique global maximizer over \(\Delta_3\). For the parameters above, the iteration gives
\(
    x^\star
    \approx
    (0.1920,\;0.2797,\;0.5282)^\top,
    \;
    \rho^\star
    \approx
    11.1941.
\)

Figure~\ref{fig:rd_dinkelbach_sim} shows that, from different initial
conditions, all allocation trajectories converge to the same optimizer
\(x^\star\), while the corresponding Dinkelbach variables converge to
\(\rho^\star\). This illustrates the convergence behavior established by
the analysis.

\begin{figure}[H]
    \centering
    \raisebox{-0.5\height}{%
        \includegraphics[width=0.4\linewidth]{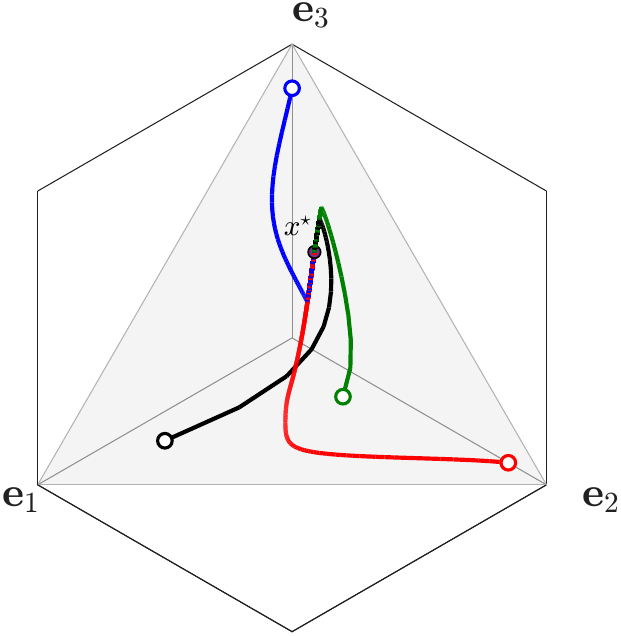}
    }
    \hspace{0.04\linewidth}
    \raisebox{-0.5\height}{%
        \includegraphics[width=0.3\linewidth]{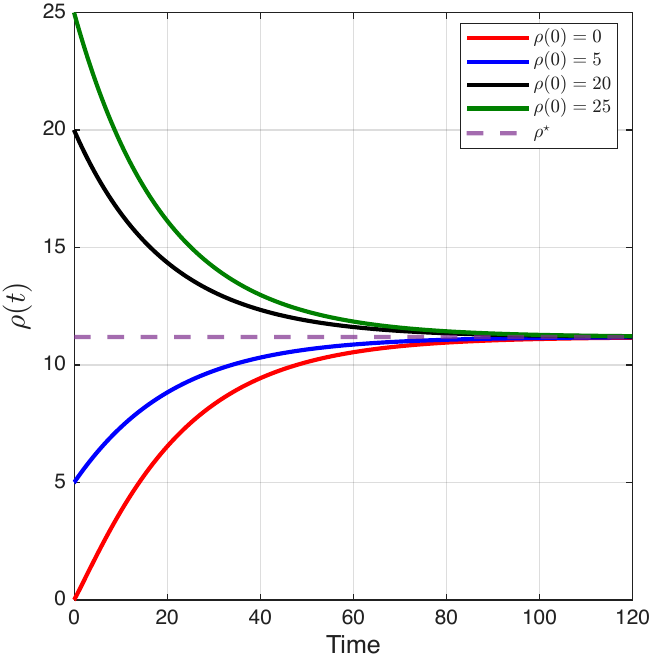}
    }
    \caption{Convergence of the replicator--Dinkelbach dynamics.}
    \label{fig:rd_dinkelbach_sim}
\end{figure}

\section{Conclusion}

This paper developed a Dinkelbach-based framework for reward-rate congestion
games. The fixed-parameter transformation recovers an exact potential
structure, yielding finite termination of the potential-level Dinkelbach
iteration under global inner maximization. Marginal externality corrections
align the potential with the social reward-rate objective. For continuous
task allocation, the proposed replicator--Dinkelbach dynamics yields
convergence of the fast and reduced subsystems and local exponential
stability of the coupled equilibrium for sufficiently slow Dinkelbach
updates. Future work will extend the framework to nonconcave transformed
potentials and to discrete-time learning dynamics with global convergence
guarantees.
\section{References}

\bibliographystyle{IEEEtran}
\bibliography{refs}

\end{document}